\documentclass{article}

\usepackage{arxiv}

\usepackage[utf8]{inputenc} 
\usepackage[T1]{fontenc}    
\usepackage[pdfencoding=auto]{hyperref}       
\usepackage{url}            
\usepackage{booktabs}       
\usepackage{amsfonts}       
\usepackage{nicefrac}       
\usepackage{microtype}      
\usepackage{lipsum}		
\usepackage{graphicx}

\usepackage{natbib}
\usepackage{multirow}
\usepackage{multicol}
\usepackage{subcaption} 
\usepackage{doi}

\usepackage{times}  
\usepackage{helvet}  
\usepackage{courier}  
\usepackage{graphicx} 
\usepackage{caption} 
\usepackage{algorithm}
\usepackage{algorithmic}

\usepackage{color}
\usepackage{amsmath,amssymb, amsthm}

\usepackage{mathtools}

\usepackage[nameinlink]{cleveref}
\providecommand{\link}[1]{[\href{#1}{link}]}
\usepackage{newfloat}
\usepackage{listings}

\DeclareCaptionStyle{ruled}{labelfont=normalfont,labelsep=colon,strut=off} 
\floatstyle{ruled}
\newfloat{listing}{tb}{lst}{}
\floatname{listing}{Listing}

\title{A Class of Topological Pseudodistances for Fast Comparison of Persistence Diagrams}

\author{ Rolando Kindelan Nu\~nez \\
        Universidad de Chile\\
	\texttt{rkindela@dcc.uchile.cl} \\
	\And
	Mircea Petrache \\
	UC Chile \\
	\texttt{mpetrache@mat.uc.cl} \\
        \And
	Mauricio Cerda \\
	Department of Electrical Engineering\\
	Universidad de Chile \\
	\texttt{mauricio.cerda@uchile.cl} \\
        \And
	Nancy Hitschfeld \\
	Universidad de Chile \\
	\texttt{nancy@dcc.uchile.cl} \\
}

\renewcommand{\undertitle}{Extended version of paper \#4518 accepted on the Main Track of AAAI24}
\renewcommand{\headeright}{AAAI24}
\renewcommand{\shorttitle}{A Class of Topological Pseudodistances for Fast Comparison of Persistence Diagrams}

\providecommand{\norm}[1]{\lVert#1\rVert}

\newtheorem{definition}{Definition}
\newtheorem{theorem}{Theorem}
\newtheorem{prop}{Proposition}
\newtheorem{lemma}{Lemma}
\newtheorem{remark}{Remark}

\hypersetup{
pdftitle={A Class of Topological Pseudodistances for Fast Comparison of Persistence Diagrams},
pdfsubject={cs.LG, math.AT, cs.CV},
pdfauthor={Rolando Kindelan, Mircea Petrache, Mauricio Cerda, Nancy Hitschfeld},
pdfkeywords={Topological Data Analysis, Persistence Diagrams, Topological Pseudodistances, Wasserstein Distance, Optimal Transport, Sliced Wasserstein Distance},
}

\begin{document}
\maketitle

\begin{abstract}
Persistence diagrams (PD)s play a central role in topological data analysis, and are used in an ever increasing variety of applications. The comparison of PD data requires computing comparison metrics among large sets of PDs, with metrics which are accurate, theoretically sound, and fast to compute. Especially for denser multi-dimensional PDs, such comparison metrics are lacking. While on the one hand, Wasserstein-type distances have high accuracy and theoretical guarantees, they incur high computational cost. On the other hand, distances between vectorizations such as Persistence Statistics (PS)s have lower computational cost, but lack the accuracy guarantees and in general they are not guaranteed to distinguish PDs (i.e. the two PS vectors of different PDs may be equal). In this work we introduce a class of pseudodistances called Extended Topological Pseudodistances (ETD)s, which have tunable complexity, and can approximate Sliced and classical Wasserstein distances at the high-complexity extreme, while being computationally lighter and close to Persistence Statistics at the lower complexity extreme, and thus allow users to interpolate between the two metrics. We build theoretical comparisons to show how to fit our new distances at an intermediate level between persistence vectorizations and Wasserstein distances. We also experimentally verify that ETDs outperform PSs in terms of accuracy and outperform Wasserstein and Sliced Wasserstein distances in terms of computational complexity.
\end{abstract}


\section{Introduction}

The processing and extraction of information from large datasets has become increasingly challenging due to the high dimensionality and noisiness of datasets. An important toolbox for describing the shape of complex data with noise robustness bounds is offered by the emerging research field of Topological Data Analysis (TDA)~\cite{Carlsson:Bulletin, EdelsbrunnerHarer2010, Cohen-Steiner2007, Cohen-Steiner2010}, which focuses on quantifying topological and geometric shape statistics of point clouds and other datasets.\par
An important advantage of TDA compared to other methods is the improved interpretability, based on insights from algebraic topology. The principal approach to encoding topological information in TDA are Persistence Diagrams (PDs) \cite{EdelsbrunnerHarer2010, zomorodian_2005} or Persistence Barcodes \cite{Ghrist2008, Carlsson:Bulletin}. TDA methods are being applied in a growing variety of fields, including time-series analysis~\cite{timeserie_tda_2016,TDAapp2017} \cite{DBLP:journals/corr/VenkataramanRT16, umeda2017}, text data analysis~\cite{WAGNER201421, tdaworddes22}, molecular chemistry \cite{TDAapp2017}, climate understanding \cite{TDAapp2017}, atmospheric data analysis~\cite{KuhnEngelkeFlatkenetal.2017, TDAapp2017}, scientific visualization~\cite{TDAapp2017}, cosmology~\cite{TDAapp2017}, combustion simulations~\cite{TDAapp2017}, computational fluid dynamics~\cite{TDAapp2017}, neurosciences~\cite{neuroscience2018}, human motion understanding~\cite{Gait17, 7844736, DBLP:journals/corr/VenkataramanRT16}, medical applications \cite{tdaretina}, volcanic eruption analysis~\cite{KuhnEngelkeFlatkenetal.2017, TDAapp2017}.
In the above TDA applications, TDA has been used as a preprocessing stage for conventional Machine Learning (ML) algorithms \cite{DBLP:journals/informaticaSI/Skraba18}, preserving interpretability, or, more rarely, as a tool to interpret the shape of clouds manipulated via Deep Learning algorithms. The overall idea is to apply persistent homology for each sample and obtain its persistence diagram. Then the space of persistence diagrams, endowed with a suitable metric or pseudometric is used as a replacement, or as an enrichment, of the original dataset.
A problem with comparison metrics between PDs is that they are computationally expensive, especially for PDs coming from $H^j$-homology with $j>0$, which are a special type of point clouds in $\mathbb R^2$. Computing the distance between two such PDs is treated as a matching problem between points in the plane, with stability and theoretical bounds based on the link with optimal transport distances like Wasserstein Distances (WDs) \cite{wasserstein1970, Cohen-Steiner2010,statwas2019,  BUBENIK2022101882}. Since WD computation for point clouds in dimension $d\geq 2$ have $O(n^3)$-complexity \cite{Munkres1957}, where $n$ is the number of points, PD comparison is often a bottleneck in ML data processing pipelines. 

\subsection{Related Work and Main Contributions}

Several approaches have arisen to make PD comparison computationally cheaper. A first direction is to optimize the precise computation of Wasserstein distance between PD's via optimization of matching problems \cite{dey2022computational, kerberacm, chen2021approximation,backurs2020scalable}, or by resorting to computationally simpler distances, such as Sliced Wasserstein Distances (SWDs) \cite{swd-carriere17a, rabinSWD2012, BonneelSWD2015, paty_subspace_2019, bayraktar_strong_2021} which roughly speaking use as distance an average of distances of $1$-dimensional projections, or to approximations of Wasserstein Distance such as Sinkhorn Distances \cite{cuturi_sinkhorn_2013, chakrabarty2021better}. Some methods use the particular geometry of PDs specifically for computing Wasserstein distances \cite{kerberacm, Khrulkov2018GeometrySA, tamaldeyWDAppr2022}. The work \cite{swd-carriere17a} applies SWD for PD comparison, but without optimizing the method towards optimum computational gains. A second direction to overcome PD comparison difficulties, is that of introducing simplified statistics via vectorization methods, with a variety of so-called Persistent Statistics (PS) \cite{henryadamsPersistenceImages2017,2022VectorizationSurvey, Chung2022}. Then distances between PS vectorizations induce pseudodistances on the originating PDs, which while computationally faster, are not guaranteed to distinguish between distinct PD's \cite{fasy2020comparing}.

In view of the above challenges of Wasserstein distance and vectorization statistics, we introduce here a class of Extended Topology pseudodistances (ETDs) between PDs which are strictly richer than PS comparison, inspired from, but much faster to compute than SWD, and which also have significant computational gains with respect to previous WD-based approximate distances. Our main contributions are the following: 
\begin{enumerate}
\item We introduce a new class of ``enhanced topology pseudodistances" (ETDs) of increasing complexity (fixable by the user), which interpolate between simple PS vectorizations and the complexity of distances such as SWD and WD. Furthermore, we verify experimentally that for real data sets the loss is minimal at low complexity, and the distinguishing power of such ETDs is comparable to the one of Wasserstein distance between PD's in applications.
\item We build the basis for a rigorous theoretical comparison of our ETD distances to present methods for computing SWD and to commonly used PS vectorization. We also prove theoretical guarantees for stability under perturbations for our distance.
\item We test our ETDs for classification applications and experimentally compare to classical methods in terms of accuracy and of computation time.
\end{enumerate}

It is worth emphasizing that, while a theoretical framework on metric comparison for PDs is well established \cite{Cohen-Steiner2007, Cohen-Steiner2010, BUBENIK2022101882, 2022VectorizationSurvey, Chung2022}, the PD construction already discards a lot of geometric and topological information about the datasets. The question of distinguishing what tasks are suited or not suited to be tackled through PD statistics is complex, and not fully settled. In the current work, we do some steps in this direction, and we hope that more research in this direction will come in the near future.

\section{Background on PDs and Their Metrics}
\subsection{A Fast Reminder on Persistence Diagrams}
We recall basic facts about PDs, see \cite{EdelsbrunnerHarer2010} for details. For  a dataset encoding as a topological space $X$, we consider a filtration $\mathcal F=\{X_t\}_{t\in[0,T]}$ in which $X_0=X$, $X_t\subseteq X_s$ for all pairs $t\leq s$ and $X_T=X$. This filtration encodes a strategy of inspection of $X$, where the precise construction algorithms for the $X_t$ depend on the task at hand and are not relevant for us. As filtration parameter $t$ increases, topological characteristics such as connected components, loops, voids, etc. appear, disappear, split or coalesce, as determined by homology classes of increasing dimension $j=0,1,2,\dots$. For each value of $j$ the increasing set of $j$-th homology groups $H_j(X_t)$ associated to $\mathcal F$ can be encoded in the so-called \emph{persistence module} of the filtration, which in high generality (via ad-hoc structure theorems) is decomposed in a direct sum of \emph{persistence intervals}, each of which allows to determine the values of time parameter $t$ at which a given homology class appears or disappears, named \emph{birth time} $b$
and \emph{death time} $d\geq b$ of the corresponding feature. The set of pairs $(b,d)\subseteq \{(x,y):\ x\leq y\}$ for $j$-dimensional homology classes form the \emph{$j$-dimensional Persistence Diagram (PD)} $PD^j(X)$ of the space $X$. For this work we will consider a fixed dimension bound $k$, and we work with the Extended Persistence Diagram (EPD) $PD(X)= \{PD^0(X), PD^1(X), \ldots, PD^k(X)\}$, in which we reiterate that $PD^j(X)$ is a collection of points in $\mathbb R^2$ for all $0\le j\le k$.

\subsection{Wasserstein-Type Distances Between PDs}
Here we recall the important metrics of interest for comparing PDs, namely geometric distances such as WD and SWD, and distances between vectorization summaries of PDs, such as PS. Consider a fixed dimension $j\ge 0$ and the PDs for dimension $j$, denoted $\mathsf P_1=PD^j(X_1), \mathsf P_2=PD^j(X_2)\subseteq \mathbb R^2$, corresponding to two datasets $X_1, X_2$. Standard comparison and theoretical guarantees such as stability under small perturbations between PDs is uses the Bottleneck Distance (BD) (cf. \cite{EdelsbrunnerHarer2010} and \cite{chazal_gromov-hausdorff_2009,chazal_persistence_2014}), which is the $p\to\infty$ limit case of $p$-Wasserstein distances:
\begin{definition}[Wasserstein distances]
Let $\mathsf P_1, \mathsf P_2\subseteq \mathbb R^2$ as above, set $\Delta:=\{(x,x):\ x\in\mathbb R\}$ and let $\Gamma$ be the set of bijections between $\mathsf P_1\cup \Delta$ and $\mathsf P_2\cup \Delta$. Then for $p\in[1,\infty)$, the \textbf{p-Wasserstein distance} between $\mathsf P_1,\mathsf P_2$ is given by
\begin{equation}\label{eq:wasserstein_p}
    W_p(\mathsf P_1, \mathsf P_2) := \left[\inf_{\gamma \in \Gamma} \sum_{d\in \mathbb P_1 \cup \Delta} \norm{d-\gamma(d)}_\infty^p \right]^{\frac{1}{p}},
\end{equation}
and the \textbf{Bottleneck distance} between $\mathsf P_1$ and $\mathsf P_2$ is given by
\begin{equation}\label{eq:wasserstein_infty}
    W_\infty(\mathsf P_1, \mathsf P_2) := \inf_{\gamma \in \Gamma} \sup_{d\in \mathsf P_1 \cup \Delta} \norm{d-\gamma(d)}_\infty.
\end{equation}
\end{definition}
The optimal algorithm for computing $W_p$ for point clouds in dimension $d\geq 2$ is the Hungarian algorithm (see \cite{kuhn1955hungarian} and Ch. 3 of \cite{cuturipeyrebook}) with complexity $O(N^3)$ if $N$ is the number of points in $\mathsf P_1\cup \mathsf P_2$. See the below discussion on time-complexity comparison for recent approximate algorithms with lower complexity. An important observation is that things improve consistently for $1$-dimensional point clouds:
\begin{prop}\label{prop:1dwass}
    For two multisets $\mathsf P_1,\mathsf P_2\subseteq \mathbb R$ the distances $W_p(\mathsf P_1, \mathsf P_2)$ can be computed in $O(N\log N)$ time.
\end{prop}
\begin{proof}[Proof sketch:] 
For distributions over $\mathbb R$ we have (see e.g. \cite[Prop. 2.16]{santambrogio2015optimal}) $W_p(\mathsf P_1,\mathsf P_2)=\|\mathsf{sort}(\mathsf P_1) - \mathsf{sort}(\mathsf P_2)\|_p$, where $\mathsf{sort}(\mathsf P)$ is the vector of coordinates of points from $\mathsf P$, in non-increasing order. Assuming that the sorting operation has complexity $O(N\log N)$ and the $\ell_p$-norm is computed with $O(N)$ operations, this gives the claimed complexity bound.
\end{proof}
Note that for $0$-dimensional homology, in important cases such as for filtrations $\mathcal F$ coming from \v{C}ech or Vietoris-Rips complexes \cite[Ch. 6]{dey2022computational}, we have birth times $b=0$ by definition, and thus $PD^0$ is a $1$-dimensional point cloud. In~\cite{horak2020topology} a \emph{topology distance (TD)} was proposed for comparing the 0-dimensional part of PD's, and it improves upon earlier statistics such as Geometry Score \cite{Khrulkov2018GeometrySA} for GAN comparison. The main difference between $W_p$ and TD is that the latter is not invariant to relabelings of the points from the PD, whereas $W_p$ is.\par
Unlike dimension $0$, PD point clouds corresponding to homology groups of dimension $j>1$ are ``truly $2$-dimensional", as birth times and death times both contain nontrivial informations about the features. As explained in the below discussion on time complexity, even considering the recent improvements on approximate Wasserstein distance computation, the cost for computing geometric distances between such PDs with good approximation, is larger than the bound from Prop. \ref{prop:1dwass}.

\subsection{PD Vectorizations and Persistence Statistics}

Vectorization is the dimension reduction of PDs from point clouds in $2$ dimensions to vector data\footnote{Note that the entries of a vector are an equivalent information to point clouds over the real line, each vector entry being identified with a coordinate, thus PD vectorizations are conceptually analogous to dimension reduction from $2$ to $1$ dimensional point clouds.}. As the projection operation loses geometric information, vectorizations inherently face the tradeoff between simplicity and informativity. For a comprehensive survey of PD vectorizations see \cite{2022VectorizationSurvey}, in which a series of vectorizations are compared in benchmark ML tasks. We focus on the best performant statistic determined in the cited paper, which turns out to be the \textbf{Persistence Statistics (PS)}. For $PD^j(X)=\{(b_i,d_i):\ i\in I_j\}$ PS includes quantile, average and variance statistics about the following collections of nonnegative numbers:
\begin{equation}\label{eq:PS}
    \{b_i\}_{ i\in I_j},\ \{d_i\}_{i\in I_j},\ \{(b_i+d_i)/2\}_{i\in I_j}, \ \{d_i-b_i\}_{i\in I_j},
\end{equation}
interpreted as, respectively, the set of birth, death, midpoints and lifetime lengths of the topological features indexed by $I_j$. Besides the above, PS includes the total number of $(b_i,d_i)$ such that $d_i>b_i$, and the entropy of the multiplicity function, whose interpretation is given in \cite{chintakunta2015entropy}. 

\subsection{Extended Topology Pseudodistance}

Our new Extended Topology Pseudodistances (ETD) are defined by projecting the PDs relative to each separate dimension $j$ over a finite set of directions, and summing the $1$-dimensional $W_p$-distances of the projections. We will apply to elements $(b_i,d_i)$ from point clouds in $\mathbb R^2$ the projection onto the $\theta$-direction defined as follows, for $\theta\in[0,2\pi)$:
\[
    \pi_\theta:\mathbb R^2\to\mathbb R,\quad \pi_\theta(x,y):=x\cos \theta + y\sin\theta.
\]
\begin{remark}\label{rmk:multiset}
    As before, for $\mathsf S\subset \mathbb R^2$, we treat $\pi_\theta(\mathsf S)$ as a \textbf{multiset} and retain the multiplicity of repeated projections.
\end{remark}
\begin{remark}
    We have $\pi_{\theta+\pi}(x)=-\pi_\theta(x)$ thus the same information is encoded in the $\pi_\theta$-projections restricted to just half of the available directions, e.g. restricting to $\theta\in[0,\pi)$.
\end{remark}
The point cloud obtained by orthogonal projection of a PD $\mathsf P^j\subset \mathbb R^2$ onto the diagonal is the following:
\begin{equation}\label{eq:projdiag}
    \widetilde P^j:=\{((b+d)/2, (b+d)/2):\ (b,d)\in \mathsf P^j\}.
\end{equation}
\begin{definition}[Extended Topology Pseudodistances]\label{def:etd}
    Let $A\subset[0,\pi)$ be a finite set of projection angles and $p\in[0,\infty]$, and consider two PDs $\mathsf P_1=PD(X_1), \mathsf P_2=PD(X_2)$ with $PD(X)$ defined as in the previous sections. For $0\le j\le k$, define the auxiliary distances 
    \[
    D_j^{A}(\mathsf P_1,\mathsf P_2):=\left(\sum_{\theta \in A} W_p\left(\pi_\theta(\mathsf P_1^j\cup \widetilde P_2^j),\pi_\theta(\mathsf P_2^j\cup \widetilde P_1^j)\right)^p\right)^{1/p},
    \]
    
    where for finite sets $\mathsf S_1,\mathsf S_2\subset \mathbb R$ of equal cardinality, we set
    \[
    W_p(\mathsf S_1, \mathsf S_2):=\|\mathsf{sort}(S_1) - \mathsf{sort}(S_2)\|_p.
    \]
        
    Then the \textbf{p-Extended Topology Pseudodistance (ETD) with projection set $A$} is defined as:
    \[
    \mathsf{ETD}_A(\mathsf P_1,\mathsf P_2):=\left(\sum_{j=0}^k D_j^A(\mathsf P_1, \mathsf P_2)^p\right)^{1/p}.
    \]
    We write $\mathsf{ETD}:=\mathsf{ETD}_{A_1}$ with $A_1=\{3\pi/4\}$ and will call this distance the \textbf{basic p-ETD}.
\end{definition}
The reason why we add the sets of the form $\widetilde P_i^j$ in the definition of $D_j^A$, is for balancing: in general $\mathsf P_1^j,\mathsf P_2^j$ do not have the same cardinality, and the correct analogue of \eqref{eq:wasserstein_p} requires including diagonal sets

Note that if the filtrations producing the PDs are such that birth values are equal to zero by construction for $\mathsf P_i^0$, then these point clouds are $1$-dimensional: we then replace $D_0^A$ by $D_0(\mathsf P_1, \mathsf P_2):=\#A\cdot W_p(\pi_{\pi/2}(\mathsf P_1),\pi_{\pi/2}(\mathsf P_2))$, i.e. consider only the ``death" coordinates, with no information loss (factor $\#A$ being introduced for normalization reasons). Also note that for $\theta=3\pi/4$ we have that $\pi_\theta(\widetilde P_i^j)=\{0\}$.

We observe that PS-distances give a strictly less informative distance than $\mathsf{ETD}_A$ due to the observation contained in the following result, whose proof is a direct computation:
\begin{lemma}\label{lem:etd-ps}
    Let $j\geq 0$ and $\mathsf P^j=\{(b_1,d_1),\dots,(b_{I_j},d_{I_j})\}$ be the $j$-dimensional PD of a dataset. Then the sets from \eqref{eq:PS} are equal to, respectively:
    \[
        \pi_0(\mathsf P^j),\ \pi_{\pi/2}(\mathsf P^j),\ \tfrac{\sqrt 2}{2}\ \pi_{\pi/4}(\mathsf P^j),\ \tfrac{2}{\sqrt 2}\ \pi_{3\pi/4}(\mathsf P^j).
    \]
\end{lemma}
In particular, the above lemma implies that $\mathsf{ETD}_{A_4}$ distance is strictly stronger than PS for $A_4:=\{0,\pi/4,\pi/2,3\pi/4\}$. Natural choices for $A\subset[0,\pi)$ with increasing numbers of elements are
\begin{equation}\label{eq:an}
    A_n:=\left\{\frac{3\pi}{4} - \frac{i}{n}\pi \ (\mathsf{mod}\pi):\ i\in\{0,\dots,n-1\}\right\}.
\end{equation}
In the above the ``$(\mathsf{mod}\pi)$" notation means that if the number $\theta_i:=3\pi/4 - i\pi/n$ becomes negative, we replace it by $\pi-\theta_i$ instead. In Appendix B we mention a more extensive list of modifications to $\mathsf{ETD}_A$ which may be useful in applications.
As noted in the proof of Prop. \ref{prop:1dwass}, we may explicitly compute $W_p$-distances from the above definition by sorting the corresponding vectors and taking $\ell_p$-norm. 
\begin{remark}[invariance properties of $\mathsf{ETD}_A$]
    In the above definition, the input of $\mathsf{ETD}_A$ are (unordered) collections of points encoded in $\mathsf P_1^j,\mathsf P_2^j, j=0,\dots,k$. In practice, we are necessarily given the point clouds in some order, and ML tasks are required to be \emph{invariant} under reordering of the points of $\mathsf P_i^j$, for all $(i,j)\in\{1,2\}\times\{0,\dots,k\}$. A wished for property of distances, adapted to ML tasks, is to actually implement this invariance, so that successive ML processing of such distances can be done without further symmetry constraints. This invariance is automatical for $\mathsf{ETD}_A(\mathsf P_1,\mathsf P_2)$ due to invariance (under relabeling of $\mathsf P_1^j$ and of $\mathsf P_2^j$) of the intermediate quantities like $W_p(\pi_\theta(\mathsf P_1^j),\pi_\theta(\mathsf P_2^j))$. 
\end{remark}
The following computational cost bounds for ETD are proved in Appendix A:
\begin{theorem}
    [Computational cost of $\mathsf{ETD}_A$]\label{thm:computcost}
    Let $\mathsf P_1,\mathsf P_2$ be two PDs corresponding to homology dimensions $0,\dots, k$, and let $A\subset [0,2\pi)$ be a set of cardinality $a$. Then the cost of calculating $\mathsf{ETD}_A(\mathsf P_1,\mathsf P_2)$ is 
    \[
    a (T_1 + (k+1)M(3 + T_2 + log M))= O(a k M \log M),
    \]
    assuming unit cost for sum or product of real numbers, and where $T_1$ is the cost to evaluate $\sin,\cos$, $T_2$ is the cost to take $p$-th powers, and $M:=\max_{0\le j\le k}(\# \mathsf P_1^j + \#\mathsf P_2^j)$.
\end{theorem}
The $M\log M$ factor in the above estimates the complexity of sorting algorithms for $M$ real numbers. Note that implementing the sorting stage with the Trimsort algorithm allows lower complexity of $O(M)$. Trimsort is a hybrid algorithm that seamlessly blends merge sort with insertion sort. It takes advantage of the inherent structure within the data to be merged, identifying sequences of pre-sorted data and integrating them into the final list, minimizing redundant comparisons.
While on the one hand ETDs can be considered as enrichments of PS-based distances (see Prop. \ref{lem:etd-ps}), the distances $\mathsf{ETD}_A$ are theoretically connected to the Sliced Wasserstein Distance (SWD). The following is a reformulation of \cite[Def. 5.1.1]{bonnotte2013unidimensional} in our setting (see also \cite[Def. 3.1]{swd-carriere17a} which is specific for PD applications and \cite{swdPhdThesis2021} for more recent advances on SWD in general):
\begin{definition}[Sliced Wasserstein Distance]
    Let $\mathsf S_1, \mathsf S_2\subseteq \mathbb R^2$ be two finite point clouds, and let $\widetilde S_i$ be the projections as in \eqref{eq:projdiag}. Then for $p\in[1,\infty]$ the \textbf{Sliced p-Wasserstein Distance (SWD)} between them is defined as
    \begin{eqnarray*}
    \lefteqn{SW_p(\mathsf S_1,\mathsf S_2)}\\
    &:=&\left(\frac{1}{\pi}\int_{0}^\pi W_p\big(\pi_\theta(\mathsf S_1 \cup \widetilde S_2),\pi_\theta(\mathsf S_2\cup \widetilde S_1)\big)^p\ d\theta\right)^{1/p}.
    \end{eqnarray*}
\end{definition}
We see that for large $n$, the set of angles $A_n$ from \eqref{eq:an} define discretizations of $[0,\pi)$ and thus we have
\begin{equation}\label{eq:limitn}
    \lim_{n\to\infty}\frac{1}{n^{1/p}}D_j^{A_n}(\mathsf P_1,\mathsf P_2)=SW_p(\mathsf P_1^j,\mathsf P_2^j).
\end{equation}
By \cite[Thm. 5.1.5]{bonnotte2013unidimensional}, for each $p\in[1,\infty)$ there exist $c_p,C_p>0$ such that restricted to pairs $2$-dimensional point clouds $\mathsf S_1, \mathsf S_2$ included in a ball of radius $\sqrt2 T$ (which is true for $\mathsf S_1=\mathsf P_1^j\cup \widetilde P_2^j$ and $\mathsf S_2=\mathsf P_2^j\cup\widetilde P_1^j$ if we truncate persistence filtrations at parameter value $T$ as in the introduction) we get the following distance comparison with Wasserstein distance:
\begin{equation}\label{eq:sww}
   c_p SW_p\leq W_p\leq C_p T^{(p-1)/3p}\left(SW_p\right)^{1/3p}.
\end{equation}
In particular, stability properties for PDs such as those proved in \cite{atienza2020stability} for $W_p$ distances, extend via \eqref{eq:sww} for $SW_p$ as well, and via \eqref{eq:limitn} we get stability bounds in the large-$n$ limit for $\mathsf{ETD}_{A_n}$. See the discussion in Appendix A. More precise quantification of these bounds at both steps ($W_p$ bounds and control for finite $n$ in \eqref{eq:limitn}) are interesting theory questions outside the scope of this paper.
\section{Theoretical Time-Complexity Comparison}

 In Table~\ref{tab:table_time}, we present the theoretical time complexity of ETD, compared to the state-of-art computation methods including Wasserstein (WD), and Sliced-Wasserstein (SWD). The WD computes Wasserstein distance using the Scikit-tda library\cite{scikittda2019} which uses a variant of Hungarian Algorithm \cite{kuhn1955hungarian}, Python Optimal Transport (Pot WD) \cite{flamary2021pot} Wasserstein is based on \cite{wsot}, Hera WD \cite{kerberacm}, and SWD \cite{swd-carriere17a}. The Pot WD, Hera WD and SWD are implemented in the Gudhi Library \cite{gudhi2014}, which is one of the most popular TDA frameworks. The recent paper \cite{dey2022computational} also compares computational cost of many recent Wasserstein approximate algorithms. PS computation requires to compute the four vectorizations \eqref{eq:PS} for each $\mathsf P_1^j, \mathsf P_2^j$ (with same notation as in Def. \ref{def:etd}), and then to take average, variance and quantiles. The vectorization calculations have a bound of $O(M)$ for each value of $j$, and the computation of statistical quantifiers requires further $O(M)$ operations for each $j$, for a total of $O(kM)$. The actual distance calculation is of lower order, and can be included in this latter bound up to increasing the implicit constant.  

\begin{table}[tb]
\centering

\begin{tabular}{|l|l|}
\hline
\multicolumn{1}{|c|}{\textbf{Distance}} & \multicolumn{1}{c|}{\textbf{Time Complexity}} \\ \hline
$\mathsf{WD}$                                      & $O(kM^3)$                                      \\ \hline
$\mathsf{Hera WD}$                                 & $O(kM^{1.5} log M)$                                  \\ \hline
$\mathsf{SWD}$                                     & $O(k\ a\ M \log{M})$                             \\ \hline
$\mathsf{PS}$                                     & $O(k M)$                             \\ \hline
$\mathsf{ETD}$                                     & $O(kM \log{M})$                          \\ \hline
$\mathsf{ETD}_{A}$                      & $O(\#A\ k\ M \log{M})$                  \\ \hline
\end{tabular}
\caption{Time complexities, where $k$ is the number of computed homology groups, $M = \max_j(\#\mathsf P_1^j+\#\mathsf P_2^j)$ and for $\mathsf{SWD}$ the quantity $a$ is the number of slices used ($a=50$ in the original implementation). Note that even for $a=\#A$, our $\mathsf{ETD}_A$ implementation is faster than the one of $\mathsf{SWD}$ with $a$ slices, because performing projections to the diagonal via our eq.(4) and Def.2 is more efficient than via the method of \cite{swd-carriere17a}.}
\label{tab:table_time}
\end{table}

\section{Experiments With PD-Based Machine Learning Tasks}

While shedding light on the underlying scalability guarantees, it is important to note that the theoretical comparison in the previous section is not the final word for practical purposes. This is due to two issues:
\begin{enumerate}
    \item Unlike for ETD, several of the most performant other methods extra data structures such as kd-trees and graphs have to be produced before the distance is computed, and these computational overheads are not explicitly discussed. 
    \item It is important to consider practical accuracy comparisons between metrics, especially for PS and ETD type metrics which may have lower distinguishing power than WD and  SWD distances on some tasks.
\end{enumerate}
We thus perform a few experiments, for comparing ETD versus state-of-art metrics based on PS, WD and SWD, to find evidence for the above two points in typical ML tasks that use PD information as an input.

We summarize in Table \ref{tab:practice_time} a wall-clock comparison between the same metrics as in Table \ref{tab:table_time}, in two applications to PDs coming from the ML pipelines, and we compare accuracy for the tasks in Table \ref{tab:knn_acc}, and Figure \ref{fig:curves} below. 

Note that indeed as expected from above point 1., there are substantial differences between Table \ref{tab:table_time} and Table \ref{tab:practice_time}. See the next sections for precise descriptions of the considered experiments in more detail.

\subsection{Experiment 1: Supervised Learning}\label{scc:exp1}

In this experiment, we perform a common TDA+ML classification using ETD, WD-based distances, and PS, via an adaptation of the experiments from \cite{2022VectorizationSurvey} to our setting.
Recall that in ~\cite{2022VectorizationSurvey} they conduct supervised learning experiments on image classification datasets: the Outex texture database \cite{outex2002}, the SHREC14 shape retrieval dataset \cite{schrec14}, and the Fashion-MNIST database \cite{fashion2017}. The cited paper follows the conventional TDA+ML hybrid classification approach, where the dataset is transformed by computing PD associated to each sample, which are then vectorized via PS or other vectorizations, followed by a classification by a conventional classifier such as Support Vector Machines.
In an adaptation of the above experiments to our framework, we use the k-Nearest Neighbor (kNN) classifier instead of Support Vector Machines, and we use as distances the ETD, WD, SWD and HeraWD distances besides PS-based distances (with the choice of exponent $p=2$ in all cases). The classifier choice was motivated mainly by the necessity of provide ad-hoc topological kernels. In ~\cite{2022VectorizationSurvey} they produce a new feature set after providing different vectorization methods, then applied directly on the conventional SVM kernels (RBF, linear). When dealing with distance matrices, defining a theoretically sound kernel is a challenging task, outside the scope of the present work. To simplify our task, we use k-Nearest Neighbor classifier which only relies on the considered distances. We compute the corresponding distances from Table \ref{tab:table_time}, and we build the kernels without the need of passing through the vectorization stage for distances other than PS. We then compute the same metrics as ~\cite{2022VectorizationSurvey} and compare the accuracy of all methods.
Focussing on the simpler case, i.e. on the Outex pattern dataset, we reduced the number of classes to classify to $10$, and we chose $20$ samples per class (some experiments on Shrec07 and Fashion-MNIST appears in Appendix \ref{app:supervised}). We computed a Cubical Complex on each class using the Gudhi library \cite{gudhi2014}. For further details and theory of cubical complexes, please consult \cite{kaczynski2004computational} as well as the following paper \cite{Wagner2012}. We compute a distance matrix using each of the distances Basic $ETD, ETDA \text{ with } A \in \{A_2,A_4,A_8,A_{16}\}, WD, Hera WD, Pot WD, SWD, PS$. We conduct a Repeated Randomized Search \cite{bergstra2012random} to determine the best $k$ and $weight$ hyperparameters for a $k$-Nearest Neighbors classifier on each distance matrices. The experimental results are shown in Table \ref{tab:knn_acc} where the weights were omitted since using the distance weight leads to optimal accuracy. As usual, the function $h(x_q)$ which k-NN uses to assign a label to a query point $x_q$, simply assigns the most voted label among its k nearest neighbors \cite{james2023introduction}:
\begin{equation}\label{eq:knnform}
h(x_q) = \arg_{y \in Y} \max \sum_{i=1}^k w(x_q, x_i) \mathbf 1 (c(x_i), y),
\end{equation}
where $c(x_i)$ is the true class of $x_i$, $w(x_q, x_i)$ is a weight and $\mathbf 1(c(x_i), y)$ is the indicator function that equals $1$ when the $x_i$ class is equal to $y$ and $0$ otherwise. We optimize over choices of $k\leq9$, and optimal values of $k$ are shown in the second column of Table \ref{tab:knn_acc}. For $w$ we tried two possible choices: $w(x_q, x_i) = 1$ (uniform) or $w(x_q, x_i) = \frac{1}{d(x_q, x_i)}$ (distance), and as shown in the last column in Table \ref{tab:knn_acc}, in all cases, for optimum $k$ the optimum choice of $w$ was the latter. We rely on Scikit-learn \cite{scikit-learn} k-NN implementation. The computation of k nearest neighbors is highly sensitive to the chosen metric, a property which allows to compare metrics on this task.\par
We see that all methods reach high accuracy, and thus this is an example of framework in which time-effectiveness of the methods would be the relevant criterion for the choice of metric. For this experiment, the third column of Table \ref{tab:practice_time} shows average time in seconds for computing the distance between two PDs in this experiment, showing that $\mathsf{ETD}_{A_1}$ would be the optimal choice.
%
\begin{table}[tb]
\centering
\resizebox{0.5\columnwidth}{!}{%
\begin{tabular}{|c|c|c|c|}
\hline
\textbf{Distance} & \textbf{Accuracy} & \textbf{k} & \textbf{\begin{tabular}[c]{@{}c@{}}weight\\ (u: uniform\\ d: distance)\end{tabular}} \\ \hline
WD & 0.98 & 6,9 & d,d \\ \hline
SWD & 0.96 & 3,6 & d,d \\ \hline
Hera WD & 0.99 & 3,4,9 & d,d,d \\ \hline
$ETD_{A_1}$ & 0.89 & 3,4 & d,d,d \\ \hline
$ETD_{A_2}$ & 0.83 & 3,4 & d,d \\ \hline
$ETD_{A_4}$ & 0.99 & 3,6 & d,d \\ \hline
$ETD_{A_8}$ & 0.99 & 3,4 & d,d \\ \hline
$ETD_{A_{16}}$ & 0.99 & 3,4 & d,d \\ \hline
$PS$ & 0.99 & 3,4 & d,d \\ \hline
\end{tabular}%
}
\caption{kNN accuracy with some of the optimal $k$ and $w$ choices for each such $k$ (see description of \eqref{eq:knnform} for detailed description).}
\label{tab:knn_acc}
\end{table}

\subsection{Experiment 2: Autoencoder Weight Topology}\label{scc:exp2}

According to \cite{NaitzatJMLR2020}, ReLU activations have a more significant impact on the topology compared to homeomorphic activations like Tanh or Leaky ReLU. ReLU activations \emph{seem to collapse} the topology in earlier layers more rapidly.\par
On the one hand, autoencoders are neural networks that aim to minimize the distance between the original data and its reconstruction, creating both an `encoder' and `decoder' \cite{GoodBengCour16}. On the other hand, the stability theorem of persistent homology \cite{Cohen-Steiner2007}, implies that training an autoencoder to reconstruct data within a narrow margin $\epsilon>0$ leads to the persistence diagrams, representing topology, that remain in close proximity within the same $\epsilon$ value. This implies that the \emph{topology cannot be altered significantly}, even when using ReLU activations and a deep autoencoder.\par
We conduct an experiment to quantify and allow interpretation to the extent to which the quantification of these properties depend on the chosen PD metrics. As a toy example we consider data sampled from two concentric balls of radiuses $1$ and $2$ in $\mathbb{R}^{100}$ as a high-dimensional dataset with 2000 sampled points each, then train a simple autoencoder with $7$ layers (of dimensions 100-20-10-3-10-20-100 respectively). After training the autoencoder, we compute persistence diagrams (corresponding to homology dimensions $0$ and $1$, i.e. for $j=0$ and $j=1$, in our notation) on the resulting point cloud given by activation vectors of each layer. Then we create so-called \textbf{topological curves} by comparing the PD of the input dataset (1st layer) with PDs corresponding to sets of output values of successive deeper layers. The comparison is done with the different metrics considered above: ${WD}_2, Hera{WD}_2, {SWD}_2$, PS-based metric and our new distances $\mathsf{ETD}_{A_i}$ for $i=1,2,4,8,16$. The topology curves are meant to assess how much each layer changes the topology. The results depend on the chosen metric for PD comparison. Results are summarized in Figures \ref{fig:aeexp} and \ref{fig:curves}.
\begin{figure}[tb]
\centering
    \includegraphics[width=\columnwidth]{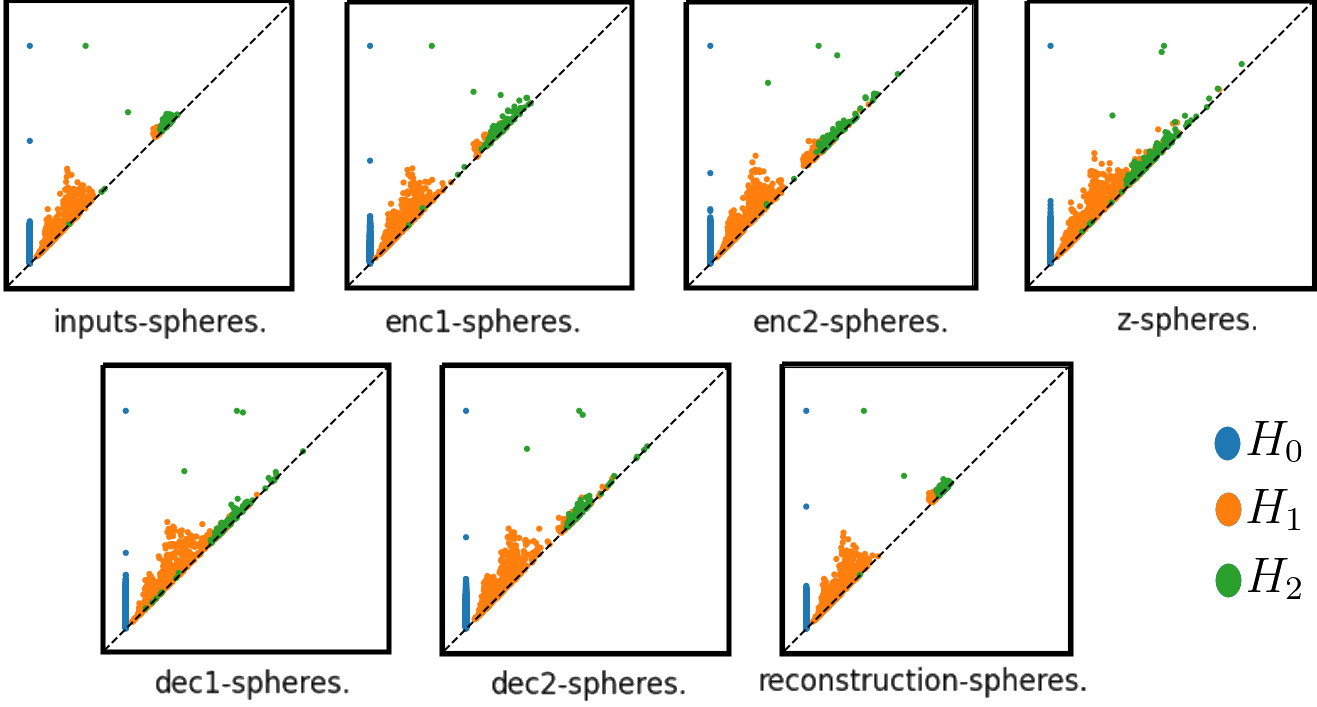}
    \caption{Example of data from Experiment 2: for each autoencoder layer, we plot the corresponding PD for $H_0, H_1,H_2$, in order from the input layer (left, first line) to the output/reconstruction layer (right, second line), for a total of $7$ layers. We plot the distance of each persistence diagram to the first one with respect to different metrics in Fig. \ref{fig:curves}.}
    \label{fig:aeexp}
\end{figure}
\begin{figure*}[tb]
 \centering
 \includegraphics[width=0.9\textwidth]{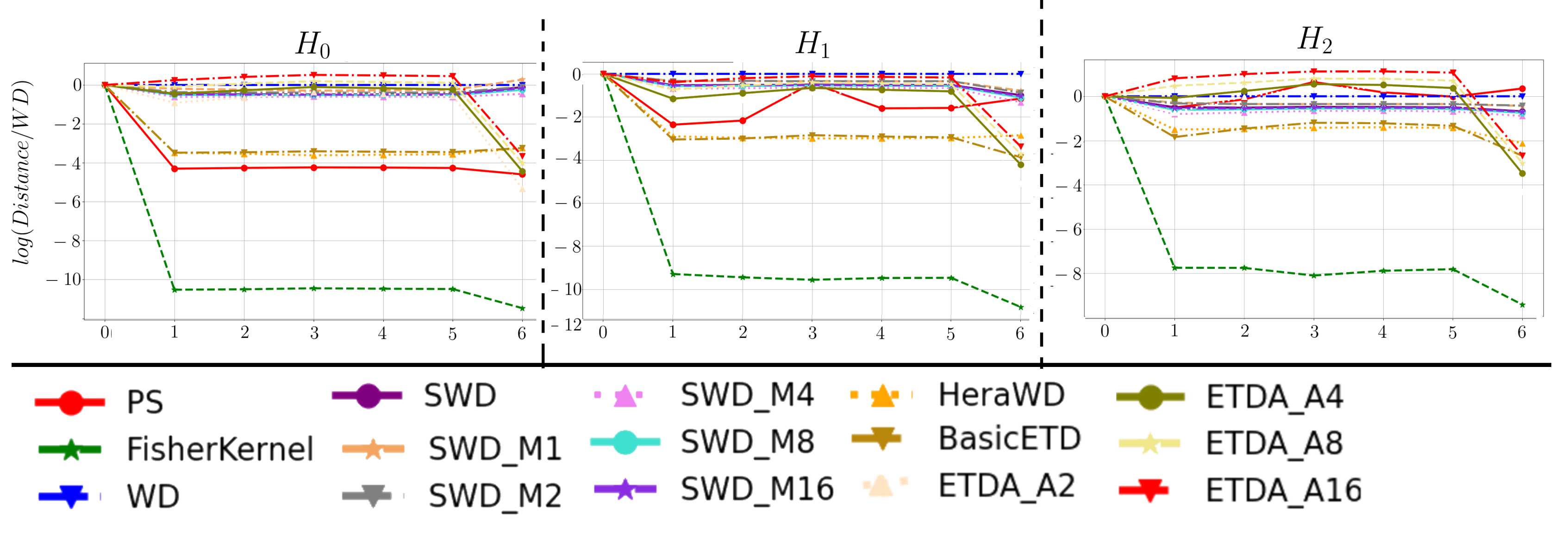}
 \caption{Example data from Experiment 2: we plot, for each homology dimension $0,1,2$, the values of $\log(\mathsf{dist}(P_i,P_0)/\mathsf{WD}(P_i,P_0)), 0\le i\le 6$ where $P_i$ is the PD of the $i$-th layer, and $\mathsf{dist}$ is amongst our allowed metrics. For completeness, we also include the Fisher Kernel distance comparison, which is much less discriminative than other metrics.}
 \label{fig:curves}
\end{figure*}
\textbf{Interpretation of the results.} Qualitative observation of topological curves as in Figure~\ref{fig:curves} across several experiments indicates that PS eliminates all variations at the level of $0$th homology group $H_0$, and introduces large variations for successive homology groups, whereas the other more precise metrics indicate lower variations. This indicates high unreliability for $PS$ metrics on qualitative tasks. We see agreement in the overall diagram shapes between $W_2$ curves, $SW_2$ curves, and $\mathsf{ETD}_{A_i}$ curves for varying values of $i$, which may be due to relative normalization factors between the metrics. Recall that as $i$ increases, in theory, due to \eqref{eq:limitn}, we expect that $\mathsf{ETD}_{A_i}$ become more accurate because it approximates SW distance more closely. We observe that the $\mathsf{ETD}_{A_i}$ curves generally diminish their oscillations as $i$ increases, with topological curve shapes similar to the one for Wasserstein distance (see Figure \ref{fig:curves}). Table \ref{tab:practice_time} shows time in seconds for computing such curves.

\begin{table}[tb]
\centering
\resizebox{0.5\columnwidth}{!}{%
\begin{tabular}{|c|lll|}
\hline
\textbf{Distance} & \multicolumn{3}{c|}{\textbf{Time in milliseconds}} \\ \cline{2-4} 
 & \multicolumn{2}{c|}{\textbf{\begin{tabular}[c]{@{}c@{}}Autoencoder\\ Weight Topology\end{tabular}}} & \multicolumn{1}{c|}{\textbf{\begin{tabular}[c]{@{}c@{}}Supervised\\  Learning\end{tabular}}} \\ \cline{2-4} 
 & \multicolumn{1}{c|}{\textbf{RELU}} & \multicolumn{1}{c|}{\textbf{LRELU}} & \multicolumn{1}{c|}{\textbf{Outex}} \\ \hline
WD & \multicolumn{1}{l|}{12544.96} & \multicolumn{1}{l|}{13626.25} & 459.24 \\ \hline
SWD & \multicolumn{1}{l|}{1588.80} & \multicolumn{1}{l|}{1551.46} & 404.09 \\ \hline
Hera WD & \multicolumn{1}{l|}{5816.86} & \multicolumn{1}{l|}{6574.51} & 864.83 \\ \hline
$ETD_{A_1}$ & \multicolumn{1}{l|}{3.69} & \multicolumn{1}{l|}{3.77} & 4.19 \\ \hline
$ETD_{A_2}$ & \multicolumn{1}{l|}{88.24} & \multicolumn{1}{l|}{72.17} & 8.55 \\ \hline
$ETD_{A_4}$ & \multicolumn{1}{l|}{118.87} & \multicolumn{1}{l|}{135.30} & 17.11 \\ \hline
$ETD_{A_8}$ & \multicolumn{1}{l|}{236.52} & \multicolumn{1}{l|}{271.12} & 34.43 \\ \hline
$ETD_{A_{16}}$ & \multicolumn{1}{l|}{469.74} & \multicolumn{1}{l|}{545.9890} & 69.11 \\ \hline
$PS$ & \multicolumn{1}{l|}{18.61} & \multicolumn{1}{l|}{17.53} & 7.51 \\ \hline
\end{tabular}%
}
\caption{Average time of each distance in milliseconds spanned by activation function and by datasets on the autoencoder and supervised learning experiments respectively.}
\label{tab:practice_time}

\end{table}

We also tested the case of $\mathsf{SWD}$ metric with a number of slices of $a=1,2,4,8,16$ versus the corresponding $ETD_{A}$ distances with $\#A=a$, and obtained that the best computational time improvement of $\mathsf{ETD}_A$ versus $\mathsf{SWD}$ is for low values of $a=\#A$: we obtain respectively an improvement by a factor of $19.4, 1.29, 1.21, 1.1, 1.1$ for the values of $a=1,2,4,8,16$ for trials of our autoencoder tasks from Experiment 2. This means that our implementation is more time-efficient than $\mathsf{SWD}$ by these factors even when applying an equal number of slices, provided this value is relatively low.

\section{Conclusion}

We have introduced a new class of distances $\mathsf{ETD}_A$ on PDs for varying small parameter set $A$. These distances on the one hand may extend the distance between vectorizations used as the basis of Persistence Statistics, and on the other hand can in theory be enriched (at the cost of increasing $A$) to approximate Sliced Wasserstein Distance between PDs. In the low-$\#A$ range considered here, $\mathsf{ETD}_A$ pseudodistance computation has theoretical complexity bounds lower than previous distances, with no additional overhead time cost (contrary to most performant WD approximations which require to construct extra data structures with an overhead to the theoretical computational cost). The cost of $\mathsf{ETD}_A$ for low number of angles $\#A$ turns out to be only marginally higher than the simpler PS-based metrics, and for $A_1$ (``basic ETD" case) it is actually substantially lower than for PS due to optimizations specific to this case (see discussion after Def. \ref{def:etd}). In practice, computational time for ETDs is considerably lower than state-of-art versions of WD or SWD distances. At the same time, in terms of accuracy loss, when tested on several common ML pipelines based on PDs, we see that ETD has higher performance than PS, and competitive accuracy performance compared to WD and SWD on ML tasks, since with ETD we reach similar qualitative description as with WD/SWD in Experiment 2, while PS-metrics seem to have unreliable qualitative behavior. Thus while having no strong theoretical guarantees of accuracy, the loss of accuracy of $\mathsf{ETD}_A$, even for for low values of $\#A$, seems to be minimal compared to finer distances such as WD or SWD distances. \par
In the comparison of our new implementation to SWD with equal numbers of slices, we find that having been careful with diagonal projections of the PD's allows notable gains for very low values of these numbers of slices (especially $a=1,2,4$) compared to SWD. These are values of major interest in our experiments, for which we find a qualitative gain of accuracy in our tasks.\par
In summary we find that $\mathsf{ETD}_A$-distances allow a compromise between low computational time and low accuracy losses, allowing to interpolate between accurate WD or SWD metrics and simple PS-based metrics, by changing an in-built complexity parameter $\#A$. Furthermore, several possibly task-specific modifications presented in Appendix B can allow further adaptation of these distances to specific tasks. We expect that the theoretical control as well as experimentation with a wider variety of tasks will be a fruitful future avenue of research.

\section*{Acknowledgements}
Rolando Kindelan Nu\~nez was supported by Beca Anid 2018/Beca doctorado Nacional-21181978, Mircea Petrache was supported by Centro Nacional de Inteligencia Artificial (CenIA) and by Fondecyt grant 1210426, Mauricio Cerda was supported by grants Fondecyt 1221696, ICN09\_015, and PIA ACT192015, Nancy Hitschfeld was supported by Fondecyt grant 1211484.

\bibliography{aaai24}

\cleardoublepage
  
\appendix
\section{Proofs of theoretical results}
\subsection{Proof of Lemma \ref{lem:etd-ps}}
It suffices to consider a point $(b_i,d_i)$: applying the projections from the statement of the lemma, we get
\begin{eqnarray*}
    \pi_0(b_i,d_i)&=&b_i,\\ 
    \pi_{\pi/2}(b_i,d_i)&=&d_i,\\ 
    \pi_{\pi/4}(b_i,d_i)&=&\frac{\sqrt2}{2}(b_i+d_i),\\
    \pi_{3\pi/4}(b_i,d_i)&=&\frac{\sqrt2}{2}(d_i-b_i).
\end{eqnarray*}
Then the proof of the lemma follows by direct rescaling and comparison to \eqref{eq:PS}.
\subsection{Proof of Theorem \ref{thm:computcost}}
In order to compute $\mathsf{ETD}_A(\mathsf P_1,\mathsf P_2)$ we need to calculate the $k+1$ distances $D_j^A(\mathsf P_1^j,\mathsf P_2^j)$ and for each of them we need to compute $\#A$ terms of the form $W_p(\pi_\theta(\mathsf P_1^j\cup \widetilde P_2^j),\pi_\theta(\mathsf P_2^j\cup\widetilde P_1^j))$. For each of these terms, the cardinality of $\mathsf P_1^j\cup \widetilde P_2^j$ and $\mathsf P_2^j\cup\widetilde P_1^j$ is bounded above by $M$ and the computation of $\pi_\theta$ then requires to compute $\sin\theta, \cos\theta$ once, and then the computation of at most $M$ terms of the form $b_i \cos \theta + d_i \sin\theta$, requiring $3$ operation each. The calculation of $W_p$ using the formula given in Def. \ref{def:etd}, requires to sort vectors of cardinality at most $M$, which takes at most $M\log M$ operations, and then to calculate the power $p$ of the $\ell_p$-norm of the difference, which takes $M$ operations of taking $p$-th power and $M-1$ operation of sum. Then we require $k(\sharp A-1)$ more sum operations and an operation of $1/p$-th power. Summing all these operations gives the first claimed bound from the theorem's statement. As $T_1, T_2$ are fixed constants, we directly get the $O(\cdots)$-based bound.

\subsection{Justification of Equation \texorpdfstring{\eqref{eq:limitn}}{xxi}}
We observe that, from Def. \ref{def:etd}, we get
\[
    \frac{1}{n} D_j^{A_n}(\mathsf P_1,\mathsf P_2)^p=\frac{1}{n}\sum_{i=0}^{n-1}W_p(\pi_{\theta_i}(\mathsf S_1), \pi_{\theta_i}(\mathsf S_2))^p.
\]
In the above, we denoted for simplicity $\mathsf S_1:=\mathsf P_1^j\cup \widetilde P_2^j$ and $\mathsf S_2:=\mathsf P_2^j\cup \widetilde P_1^j$, and also $A_n=\{\theta_i:\ 0\le i<n\}$. Note that due to the definition \eqref{eq:an} of $A_n$ the values $\theta_i$ are $n$ equally spaced points in the interval $[0,\pi)$. It is then direct to check that the last displayed formula is actually a Riemann sum approximation of
\[
    \frac{1}{\pi}\int_0^\pi W_p(\pi_\theta(\mathsf S_1),\pi_\theta(\mathsf S_2))^p d\theta = SW_p(\mathsf P_1^j,\mathsf P_2^j)^p,
\]
and the classical Riemann integral approximation theorem gives the desired result.

\subsection{Details about stability control with \texorpdfstring{$\mathsf{ETD}_{A_n}$}{ETD-An}}
The classical result \cite{chazal_gromov-hausdorff_2009, cohen2010lipschitz} is that if for $X_1,X_2$ finite metric spaces, the $j$-dimensional PDs $\mathsf P_1,\mathsf P_2$ are obtained from the Rips complex, then 
\[
    W_\infty(\mathsf P_1,\mathsf P_2) \leq GH(X_1,X_2),
\]
in which $GH$ is the Gromov-Hausdorff distance (see e.g. \cite[Ch.7]{burago2022course}). Next, from e.g. \cite[Lem. 3.5]{atienza2020stability} we have the observation, based on a similar result for $\ell_p,\ell_\infty$ norms, that $W_p$ is controlled by $W_\infty$, however this is comes at the cost of including a dependence on $M:=\#\mathsf P_1+\#\mathsf P_2$, for the most interestig lower bound (we include the upper bound for completenes but we do not use it):
\[
    M^{-1/p}W_p(\mathsf P_1,\mathsf P_2)\leq W_\infty(\mathsf P_1,\mathsf P_2)\leq W_p(\mathsf P_1,\mathsf P_2).
\]
Using \eqref{eq:sww} in conjunction with the above lower bounds, we then obtain a stability control on $SW_p$ as well, and via \eqref{eq:limitn} we obtain bounds for $\mathsf{ETD}_{A_n}$ for large $n$. Note that at the moment these bounds are somehow weak, due to the limit in \eqref{eq:limitn} and to the dependence on $M$ in the above $W_p / W_\infty$ bound. Thus we focused on experimental comparisons of $\mathsf{ETD}_A$ for small values of $\#A$, and leave this theoretical comparison to future work.

\section{Possible ad-hoc modifications for ETDs}
In applications, there may be a series of minor modifications to the initial ETD definition (Def. \ref{def:etd}) that users may want to implement for specific tasks. We list a list of such possible modifications here.
\subsection{Dimensional reweighting}
It may be useful to include ad-hoc normalization weights $w=(w_0,\dots,w_k)$ with $w_j>0$, for the auxiliary distances $D_j^A, j=0,\dots,k$ in Def. \ref{def:etd}, and define 
\[
    \mathsf{ETD}_A^w(\mathsf P_1,\mathsf P_2):= \left(\sum_{j=0}^k w_j D_j^A(\mathsf P_1^j, \mathsf P_2^j)^p \right)^{1/p}.
\]
The usefulness of such weights (which can be hyperparameters or learned parameters in a ML task) depends of relative relevance of different types of topological features for given tasks.

\subsection{Dimension-dependent angle choices}
It may be useful to let the choice of angle sets $A$ depend upon dimension $j$, rather than being fixed as in Def. \ref{def:etd}. This means that we fix sets $A_0,\dots,A_j\subset [0,2\pi)$ and in Def. \ref{def:etd} replace $D_j^A$ by 
\[
    D_j^{A_j}(\mathsf P_1^j,\mathsf P_2^j):=\left(\sum_{\theta\in A_j}W_p(\pi_\theta(\mathsf P_1^j),\pi_\theta(\mathsf P_2^j))^p\right)^{1/p}.
\]
A special case of these changes is mentioned immediately after Def. \ref{def:etd}: if $\mathsf P^0$ has by construction birth coordinates $b=0$, then it is useful to only consider the death coordinates $d$, i.e. to use $A_1=\{\pi/2\}$. The natural choices for $A_j$ are given in \eqref{eq:an}.

\subsection{Randomization}
A common strategy for dimension reduction endeavors such as the passage from WD to SWD often profit from the use of randomized projections, akin to the framework of the Johnson-Lindenstrauss theorem. This has been implemented in \cite{deshpande2019max, nietert2022statistical,nadjahi_fast_2021} for classical SWD optimization for point clouds in large dimension. Choosing randomized projection for the dimension reduction from $2$ to $1$ dimensional multisets in our framework, will give minimal computational time gains compared no fixing the $A_n$ as above. Furthermore, theoretical convergence rate in the approximation to $SW_p$ like \eqref{eq:limitn} is reduced for randomized $A_n$.

\subsection{1-dimensional Wasserstein alternatives with general metrics replacing \texorpdfstring{$\ell_p$}{Lp}-norm}

In some applications such as word2vec vectorizations, other metrics for comparing vectorizations are used, besides classical Banach space norms. This can be adapted to the present framework. For example, to compare accuracy values for $A=A_n$ as in \eqref{eq:an}, besides using $EDT_A$ with based on $W_p(\mathsf S_1,\mathsf S_2)=\|\mathsf{sort}(\mathsf S_1) - \mathsf{sort}(\mathsf S_2)\|_{\ell_p}$ as in Definition \ref{def:etd}, one may try out different distances instead of $\ell_p$, in order to compare the vectors $V_i:=\mathsf{sort}(\mathsf S_i)$, for example cosine similarity corresponding to taking a scalar product of the vectors:
    \[
        d_{\cos}(V_1,V_2):=\left|\sum_i (V_1)_i (V_2)_i\right|.
    \]
We sum these distances over $\theta\in A_n$ as in Definition \ref{def:etd} as an alternative to our $\mathsf{EDT}_{A_n}$. While the mathematical grounding for these alternative definitions is less strong than for the canonical version in Def. \ref{def:etd}, it may be worth to try such alternatives for experimental evaluation on ad-hoc tasks.
\section{Implementation details}

\begin{multicols}{2}
\begin{algorithm}[H]
    
    \caption{$get\_L\_vector$: Computes the $\mathcal{L}$ vector.}
    \label{alg:get_L_vector}
    \textbf{Input}: A persistence diagram $\mathsf P$.\\
    \textbf{Output}: The longevity vector $L$ of $\mathsf P$.
    
    \begin{algorithmic}[1]     
    \STATE $k \leftarrow |\mathsf P|$, $L \leftarrow ()$,  $j \leftarrow 0$
    
    \FOR{$j  \leq k$}
        \STATE $L_j \leftarrow ()$
        \FOR{$(b,d)  \leq \mathsf P^j$}
            \STATE $L_j.add(d - b)$
        \ENDFOR 
        
        \STATE $L_j \leftarrow \mathsf{sort}(L_j)$ 
        \STATE $L.add(L_j)$ 
        \STATE $j \leftarrow j + 1$
    \ENDFOR
    \RETURN $L$
    \end{algorithmic}
\end{algorithm}
 
\begin{algorithm}[H]
    
    \caption{$get\_V\_vector$: Computes the $\mathcal{L}$ vector.}
    \label{alg:get_V_vector}
    \textbf{Input}: A persistence diagram $\mathsf P$ and an angle $\alpha$.\\
    \textbf{Output}: The $V$ vector of $\mathsf P$ associated to $\alpha$.
    
    \begin{algorithmic}[1]     
    \STATE $k \leftarrow |\mathsf P|$, $V \leftarrow ()$, $j \leftarrow 0$
    \STATE $(b_\alpha, d_\alpha) \leftarrow (cos(\alpha), sin(\alpha)) $
    
    \FOR{$j  \leq k$}
        \STATE $V_j \leftarrow ()$
        \FOR{$(b,d)  \leq \mathsf P^j$}
            \STATE $v_j \leftarrow b \cdot b_\alpha + d \cdot d_\alpha$    
            \STATE $V_j.add(v_j)$
        \ENDFOR 
        
        \STATE $V_j \leftarrow \mathsf{sort}(V_j)$ 
        \STATE $V.add(V_j)$ 
        \STATE $j \leftarrow j + 1$
    \ENDFOR   
    \RETURN $V$
    \end{algorithmic}
\end{algorithm}
\end{multicols}

\begin{algorithm}[H]
    
    \caption{$BasicETD$: Computes the $ETD_{A_1}$ pseudodistance.}
    \label{alg:get_basic_etd}
    \textbf{Input}: Two persistence diagrams $\mathsf P_1, \mathsf P_2$ and a power value $p$.\\
    \textbf{Output}: The basic ETD pseudodistance between $\mathsf P_1, \mathsf P_2$.
    
    \begin{algorithmic}[1]     
    \STATE $L_1 \leftarrow get\_L\_vector(\mathsf P_1)$
    \STATE $L_2 \leftarrow get\_L\_vector(\mathsf P_2)$
    \STATE $L_1',L_2' \leftarrow padding\_zeroes(L_1, L_2)$ 
    \STATE $k \leftarrow |L_1'|$
    \RETURN $\norm{L_1' -L_2'}_p$
    \end{algorithmic}
\end{algorithm}

\begin{algorithm}[H]
    
    \caption{$ETD_A$: Computes the $ETD_A$ pseudodistance.}
    \label{alg:get_etd_A}
    \textbf{Input}: Two persistence diagrams $\mathsf P_1, \mathsf P_2$, and angle set $A$ and a power value $p$.\\
    \textbf{Parameter}: Optionally $\widetilde P_1, \widetilde P_2$ can be precomputed. \\
    \textbf{Output}: The $ETD_A$ pseudodistance between $\mathsf P_1, \mathsf P_2$ considering $A$.
    
    \begin{algorithmic}[1]
    \STATE $E \leftarrow ()$
    
    \IF{ $\widetilde P_1 = \emptyset$ }
        \STATE $\widetilde P_1^j \leftarrow \{((b+d)/2, (b+d)/2):\ (b,d)\in \mathsf P_1^j\}$
    \ENDIF
    \IF{ $\widetilde P_2 = \emptyset$ }
        \STATE $\widetilde P_2^j \leftarrow \{((b+d)/2, (b+d)/2):\ (b,d)\in \mathsf P_2^j\}$
    \ENDIF
    \STATE $\mathsf P_1.extend(\widetilde P_2)$ \COMMENT{dimension-wise concatenation}
    \STATE $\mathsf P_2.extend(\widetilde P_1)$
    
    \FOR{$\alpha \in A$}
    \STATE $V_1 \leftarrow get\_V\_vector(\mathsf P_1, \alpha)$
    \STATE $V_2 \leftarrow get\_V\_vector(\mathsf P_2, \alpha)$
    \STATE $etd_\alpha \leftarrow \norm{V_1 - V_2}_p^p$
    \STATE $E.add( etd_\alpha)$
    \ENDFOR
    \RETURN $\norm{E}_p$
    \end{algorithmic}
\end{algorithm}

Algorithm \ref{alg:get_basic_etd} and Algorithm \ref{alg:get_etd_A} depends on the implementation of $L, V$ vectors and projections to the diagonal. We used a naive implementation of these methods in our experiments, but these results can be enhanced using more sophisticated libraries Numpy, PyTorch, TensorFlow, GPU kernels or python multiprocessing libraries. To demonstrate this potential, we have developed a prototype implementation utilizing Numpy \cite{harris2020array}, which achieves considerable speed ups in computing these functions, and therefore the ETDs. 

\begin{multicols}{2}
    
\begin{lstlisting}
import numpy as np
    
def get_L_vector(PD):
    longevity = []
    hgroups = len(PD)
    for d in range(hgroups):        
        # persistence computation
        dlongevity = np.diff(PD[d], axis=1).T
        # numpy descending sorting
        dlongevity[0][::-1].sort()

        longevity.append(dlongevity[0])

    return longevity

def get_V_vector(PD, alpha):
    if alpha is None:
        return ExtendedTopologyDistanceHelper.get_L_vector(PD)
    
    alpha_interval = np.array([np.cos(alpha), np.sin(alpha)]) 
    longevity = []
    hgroups = len(PD)

    for d in range(hgroups):
        dlongevity = PD[d] @ alpha_interval
        dlongevity[::-1].sort()
        longevity.append(dlongevity)

    return longevity
    
def get_maximum_sizes(PDs):
    max_sizes = {}
    for PD in PDs:
        for j, PDj in enumerate(PD):
            hgcard = len(PDj)
            if d not in max_sizes: # missing dimension
                max_sizes.update({d: hgcard})
            if max_sizes[d] < hgcard:
                max_sizes[d] = hgcard
    return max_sizes
        
def padding_zeroes(PD1, PD2):
    max_sizes = get_maximum_sizes([PD1, PD2])

    AA = []
    BB = []
    siA = len(A)
    siB = len(B)

    for d in max_sizes:
        padded = np.zeros(max_sizes[d])
        padded2 = np.zeros(max_sizes[d])
        if siA > d:
            padded[:len(A[d])] = A[d]
        if siB > d:
            padded2[:len(B[d])] = B[d]

        AA.append(padded)
        BB.append(padded2)

    return AA, BB

def get_projection(PD):
    mvect = []
    max_dim = len(PD)
    for d in range(max_dim):
        mid = np.sum(PD[d], axis = 1)*0.5

        mvect.append(np.vstack((mid, mid)).T)

    return mvect

\end{lstlisting}
\end{multicols}

In \cite{harris2020array} there is an exposition on the diverse Numpy capabilities for array programming (including a) data structures, b) indexing, c) vectorization, d) broadcasting and e) reductions (illustrated in Figure 1 of \cite{harris2020array}). Its foundational role underpins numerous numerical computation libraries, and its influence extends to the core of Machine Learning and Deep Learning tools such as PyTorch and TensorFlow, which adopt a Numpy-like syntax. This compatibility ensures a seamless transition of code between Numpy and these frameworks, enabling almost direct mapping of implementations.\par


The benefits of adopting these Numpy-enhanced functions are evident in our extended experiments, which are documented with the `np' prefix in Table \ref{tab:knn_time} within Appendix \ref{app:supervised}, and in Table \ref{tab:autoencoder} found in Appendix \ref{app:autoencoder}. The source code of this paper is available in \href{http://github.com/rolan2kn/aaai2024_etd_src}{http://github.com/rolan2kn/aaai2024\_etd\_src.}

\section{Extended Supervised Learning experiment}\label{app:supervised}

We extends our experiments of Section \ref{scc:exp1} on Supervised Learning to Shrec07 \cite{datasetShrec07} and Fashion-MNIST \cite{xiao2017fashionmnist}. 

\subsection{Dataset Description}

The aim of the Shrec07 dataset \cite{datasetShrec07} was to provide a common benchmark for the evaluation of the effectiveness of 3D-shape retrieval algorithms. The dataset consists of 400 OFF format files, encoding meshes of 3D surfaces without defective holes or gaps, subdivided into 19 classes of 20 elements each, shown in Figure \ref{fig:shrec2007}. Note that the meshes have highly varying sizes, even within the same class. As explained in Section \ref{sec:Dresults}, we use Shrec07 in a supervised learning task by splitting the dataset into train and test sets. \par


Fashion-MNIST \cite{xiao2017fashionmnist} is a dataset consisting of 28x28 grayscale article images, labelled according to 10 possible classes, and origined from a fashion retailer's article images. The dataset  consists of a training set of 60000 examples and a test set of 10000 examples. \par

\begin{figure}
     \centering
    \begin{subfigure}[t]{0.49\textwidth}
        \includegraphics[width=\textwidth]{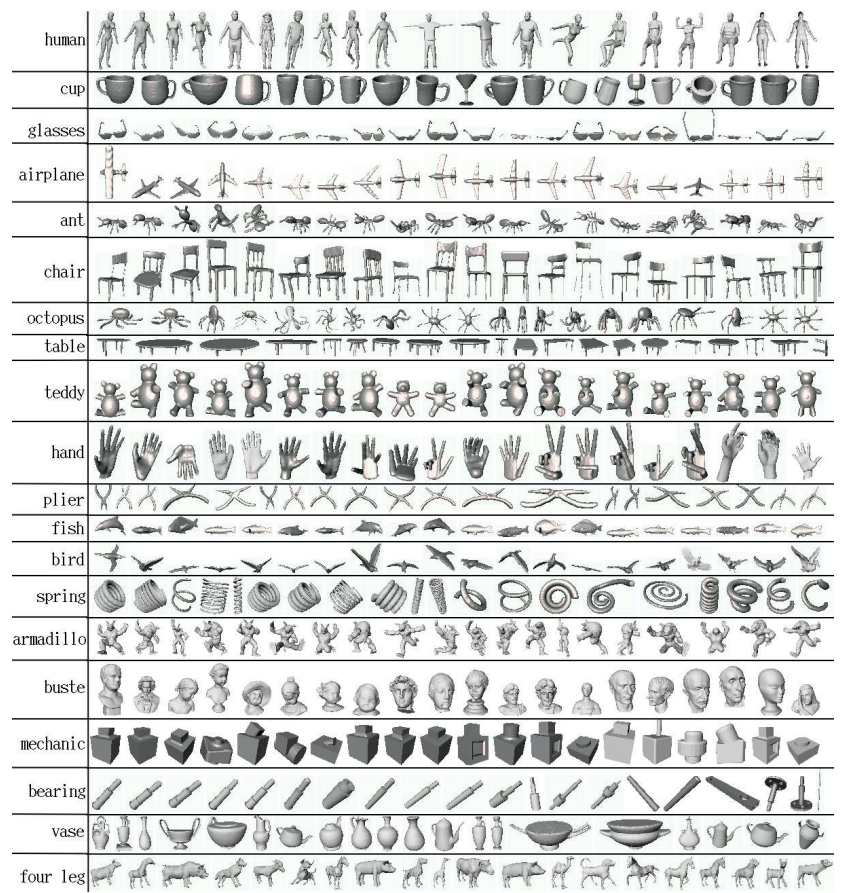}
        \caption{The Shrec'07 dataset from  \cite{datasetShrec07}. Note that compared to the original case, the "spring" class (row 7 from the bottom) was removed in our experiment, as the examples in this class are not topologically consistent.}
    \label{fig:shrec2007}
    \end{subfigure}
    \hfill
    \begin{subfigure}[t]{0.49\textwidth}
        \includegraphics[width=\textwidth]{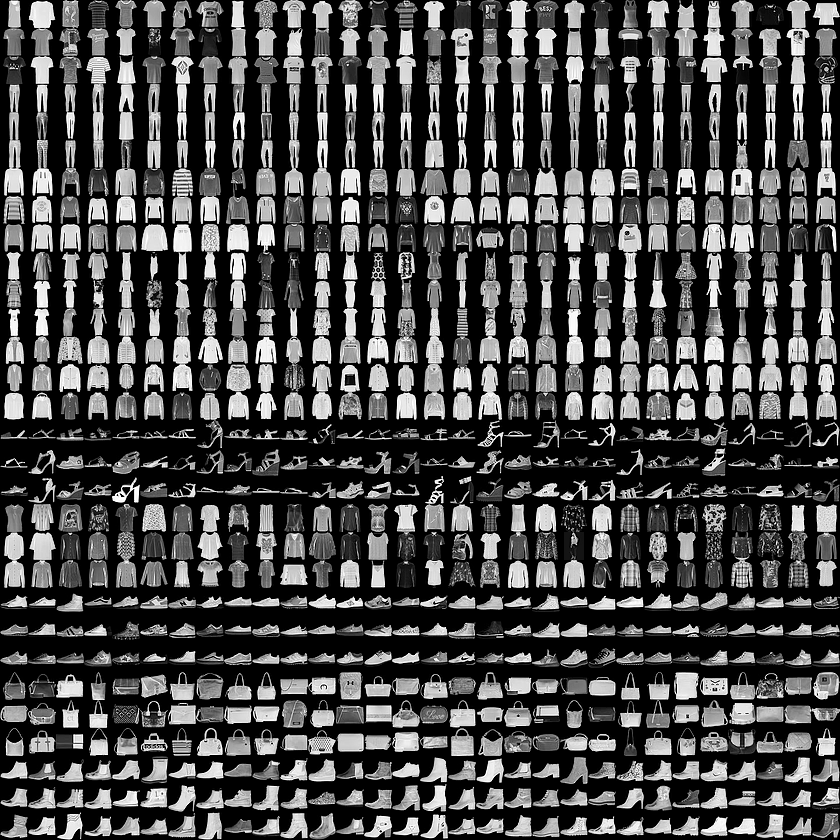}
        \caption{The fashion mnist dataset divided into class examples.}
    \label{fig:fashion}
    \end{subfigure}
    
    \caption{Datasets}\label{fig:alldatasets}
\end{figure}

\subsection{Data preprocessing}

For the Shrec07 dataset, we compute persistence diagrams on each mesh-sample and create a new dataset composed with obtained the persistence diagrams. Then, using our proposed distances and kernels, we assess the classifiers. Each Shrec07 example is a mesh object stored in an OFF (Object File Format) file, from which we extract the 3D points ignoring edges and faces. Then we compute a 3D Heat Kernel Signature (HKS) \cite{Zobel2011GeneralizedHK} and produce a 6-dimensional embedding by concatenating HKS coordinates with mesh vertex coordinates. Then we build an associated sparse Rips complex up to $H_2$. In order to get possibly better results, one could have considered a filtration via different HKS's depending of a parameter value: we did not do this due to time constraints, as it would be computationally intensive.\par
For the Fashion-MNIST setting, we adopt a preprocessing stage where on every image we apply the Histogram of Oriented Gradients (HOG) operator \cite{hog2005}. The process includes global image normalization, computing the gradient image in row and col (8 orientations), computing gradient histograms (8x8 pixels per block), normalizing across blocks (2x2 cells per block), and flattening image $i$ into a feature vector $i_f$. Finally the obtained feature vector is normalized as $i_s = \frac{i_f - mean(i_f)}{\sigma(i_f)}$.   


\subsection{Results}\label{sec:Dresults}

We formulate the Shrec07 dataset classification task as follows. We randomly select 15 examples per class, for a total of 285 examples, and for the test set we sample 10 of the 95 remaining files.\par 
On Fashion-MNIST dataset, we use 1000 samples as train set (100 per class) and again 10 samples as test set. Table \ref{tab:eknn_acc} shows the accuracy obtained with these tasks. Table \ref{tab:knn_time} shows the average time in milliseconds that takes compute one of such distances. We use a laptop Acer Aspire A315-42 with 32 GB RAM, and an AMD Radeon vega 10 graphics. 

\begin{table}[H]
\centering
\begin{tabular}{|c|cc|cc|}
\hline
\multirow{2}{*}{\textbf{Distance ($H_1$)}} & \multicolumn{2}{c|}{\textbf{Shrec07}}               & \multicolumn{2}{c|}{\textbf{Fashion-MNIST}}          \\ \cline{2-5} 
                                   & \multicolumn{1}{c|}{\textbf{Accuracy}} & \textbf{k} & \multicolumn{1}{c|}{\textbf{Accuracy}} & \textbf{k} \\ \hline
$WD$           & \multicolumn{1}{c|}{0.9}  & 1-6,10,13 & \multicolumn{1}{c|}{0.6} & 8,10,14-17,22,23 \\ \hline
$Hera\_WD$     & \multicolumn{1}{c|}{0.5}  &           & \multicolumn{1}{c|}{0.7} & 6,10             \\ \hline
$SWD_{a=1}$    & \multicolumn{1}{c|}{0.6} & 2-4       & \multicolumn{1}{c|}{0.4} & 1,2,8-29         \\ \hline
$SWD_{a=2}$    & \multicolumn{1}{c|}{0.6} & 2-4       & \multicolumn{1}{c|}{0.5} & 4,11,            \\ \hline
$SWD_{a=4}$    & \multicolumn{1}{c|}{0.6} & 2,9,10    & \multicolumn{1}{c|}{0.6} & 2,3,5,8-11,      \\ \hline
$SWD_{a=8}$    & \multicolumn{1}{c|}{0.6} & 2,6-14    & \multicolumn{1}{c|}{0.7} & 6                \\ \hline
$SWD_{a=16}$   & \multicolumn{1}{c|}{0.5} & 1-4,9,10,12  & \multicolumn{1}{c|}{0.6} & 3,6              \\ \hline
$ETD_{A_1}$    & \multicolumn{1}{c|}{0.8} & 1-4       & \multicolumn{1}{c|}{0.7} & 4,5              \\ \hline
$ETD_{A_2}$    & \multicolumn{1}{c|}{0.8} & 1-4       & \multicolumn{1}{c|}{0.7} & 4,5,             \\ \hline
$ETD_{A_4}$    & \multicolumn{1}{c|}{0.8} & 1,7,8     & \multicolumn{1}{c|}{0.7} & 8-29             \\ \hline
$ETD_{A_8}$    & \multicolumn{1}{c|}{0.9} & 2,7,8       & \multicolumn{1}{c|}{0.7} & 6,10-13,16-29    \\ \hline
$ETD_{A_{16}}$ & \multicolumn{1}{c|}{0.9} & 7,8    & \multicolumn{1}{c|}{0.6} & 8-29            \\ \hline
$PS$           & \multicolumn{1}{c|}{0.4} & 3,4       & \multicolumn{1}{c|}{0.5} & 9,12,13,21,29 \\ \hline
Fisher Kernel
& \multicolumn{1}{c|}{0.5}    &  5,8,10         & \multicolumn{1}{c|}{0.4} & 24-41            \\ \hline
\end{tabular}
\caption{Accuracy of datasets using KNN with the respective k values. }
\label{tab:eknn_acc}
\end{table}

\begin{table}[H]
\centering
\begin{tabular}{|c|cl|}
\hline
\multirow{3}{*}{\textbf{Distance}} & \multicolumn{2}{c|}{\textbf{Time in milliseconds}}            \\ \cline{2-3} 
               & \multicolumn{2}{c|}{\textbf{Supervised Learning}} \\ \cline{2-3} 
                                   & \multicolumn{1}{c|}{\textbf{Shrec07}} & \textbf{Fashion-MNIST} \\ \hline
$WD$           & \multicolumn{1}{c|}{305774.22}     & 4794.9476    \\ \hline
$Hera\_WD$     & \multicolumn{1}{c|}{46372.22}      & 4042.86      \\ \hline
$SWD_{a=1}$    & \multicolumn{1}{c|}{127.26}        & 17.74        \\ \hline
$SWD_{a=2}$    & \multicolumn{1}{c|}{199.33}        & 386.58       \\ \hline
$SWD_{a=4}$    & \multicolumn{1}{c|}{331.42}        & 48.54        \\ \hline
$SWD_{a=8}$    & \multicolumn{1}{c|}{617.25}        & 90.59        \\ \hline
$SWD_{a=16}$   & \multicolumn{1}{c|}{1184.80}       & 173.90       \\ \hline
$ETD_{A_1}$    & \multicolumn{1}{c|}{12.34}         & 1.48         \\ \hline
$ETD_{A_2}$    & \multicolumn{1}{c|}{132.19}        & 24.68        \\ \hline
$ETD_{A_4}$    & \multicolumn{1}{c|}{251.71}        & 46.78        \\ \hline
$ETD_{A_8}$    & \multicolumn{1}{c|}{490.11}        & 90.79        \\ \hline
$ETD_{A_{16}}$ & \multicolumn{1}{c|}{967.94}       & 178.66       \\ \hline
$npETD_{A_1}$    & \multicolumn{1}{c|}{0.77}         & 0.06         \\ \hline
$npETD_{A_2}$    & \multicolumn{1}{c|}{8.36}        & 0.62        \\ \hline
$npETD_{A_4}$    & \multicolumn{1}{c|}{17.67}        & 1.10        \\ \hline
$npETD_{A_8}$    & \multicolumn{1}{c|}{35.97}        & 2.14        \\ \hline
$npETD_{A_{16}}$ & \multicolumn{1}{c|}{72.95}       & 4.24       \\ \hline
$PS$           & \multicolumn{1}{c|}{42.3}         & 8.30         \\ \hline
Fisher Kernel   & \multicolumn{1}{c|}{12724.06}              & 454.71       \\ \hline
\end{tabular}
\caption{Average times in milliseconds of a single distance computation in $H_1$.}
\label{tab:knn_time}
\end{table}

As future work, in the case of Fashion-MNIST dataset, we plan to explore more sophisticated frameworks for classifying grayscale images such as \cite{Bergomi2019} and \cite{garin2019topological}. Both works apply different cubical filtrations to images to generate a wide range of topological features that leads to similar accuracy to more complex scenarios but with much lower number of features \cite{Bergomi2019}. These works were focused on applying vectorization methods included in \cite{2022VectorizationSurvey}, which prevents them to be a perfect match for our tasks since we are operating directly on the persistence diagrams considering different distances. However, in \cite{garin2019topological}, the authors use Wasserstein and Bottleneck amplitude, computed by calculating the $L_2$ norm of the persistence intervals, that is the L2-norm of unsorted vector $L$, multiplied by $\sqrt{2}/2$, similar to the $3\frac{3\pi}{4}$ projection) which is strictly less discriminative than our ETD, suggesting that we could obtain similar or better results. 

\section{Extended autoencoder experiment}\label{app:autoencoder}

The concentric spheres are manifolds of dimension 99, and thus if densely sampled we expect to find 2 components at 0-dimensional homology level, and trivial homology in dimensions 1 and 2. As we use a sparse random sampling by of 2000 points, the induced topological errors are nontrivial. In Figures \ref{fig:relu} and \ref{fig:lrelu}, we illustrate the latent spaces captured by our autoencoder and further elaborate on the Persistence Diagrams (PDs) for each layer, calculating up to the second homology group, $H_2$. These point clouds, being considerably dense, These point clouds, being considerably dense, made the calculation of persistence diagrams in dimensions greater than 1 a challenge, requiring the use of complex dispersed Rips. This was achieved by setting the 0.75 quantile of the Manhattan distance matrix as the upper limit for edge length, employing a sparsity factor of 0.5, and performing edge collapses primarily on the 1-skeleton before expanding to higher dimensions up to dimension $3$ to be sure we capture topological features in $H_2$. Additionally, in Table \ref{tab:autoencoder}, we present the average computation times, expressed in milliseconds, for calculating a single distance. This includes a comparative analysis of the average times for both the $SWD$ and the $ETD$, utilizing an identical count of projections.


\begin{figure}[H]
    \centering
    \includegraphics[width=1\columnwidth]{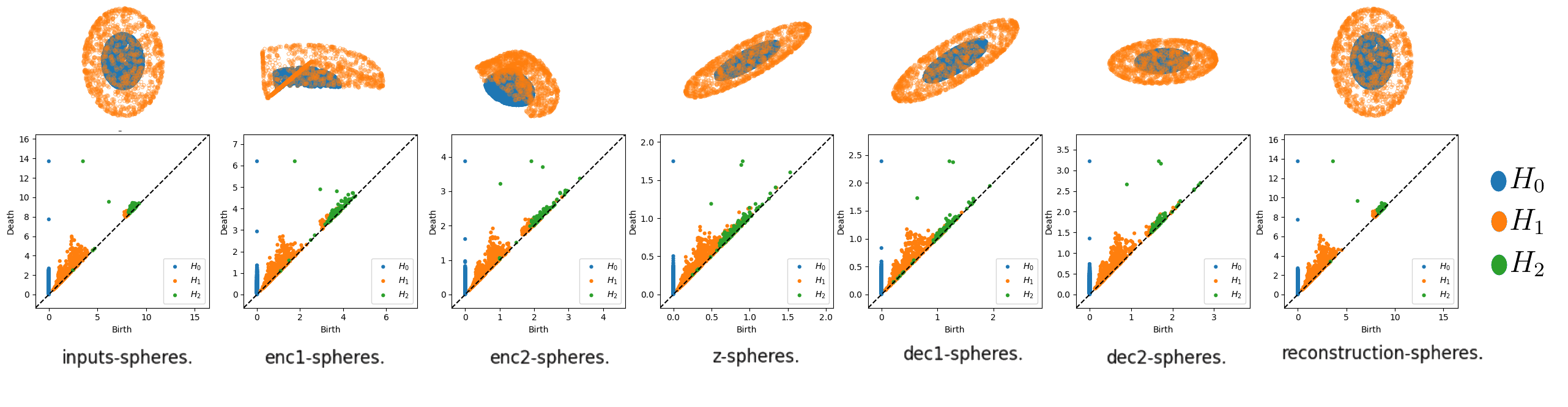}
    \caption{The Relu latent space information and persistence diagrams up to $H_2$.}
    \label{fig:relu}
\end{figure}

\begin{figure}[H]
    \centering
    \includegraphics[width=1\columnwidth]{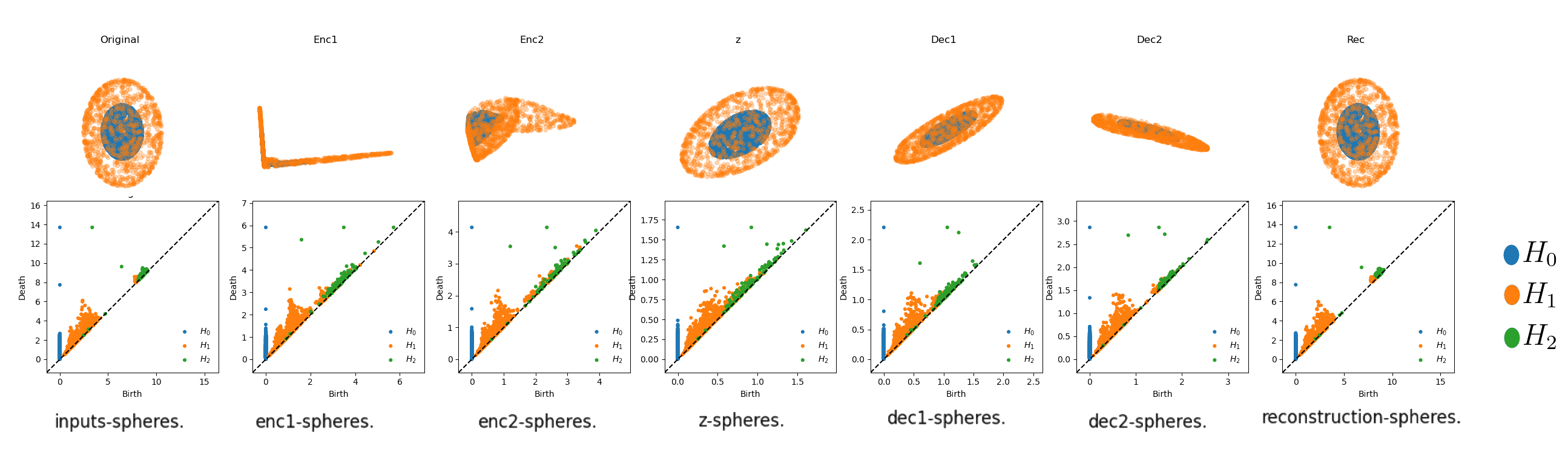}
    \caption{The LRelu latent space information and persistence diagrams up to $H_2$.}
    \label{fig:lrelu}
\end{figure}

\begin{figure}[H]
    \centering
    \includegraphics[width=1\textwidth]{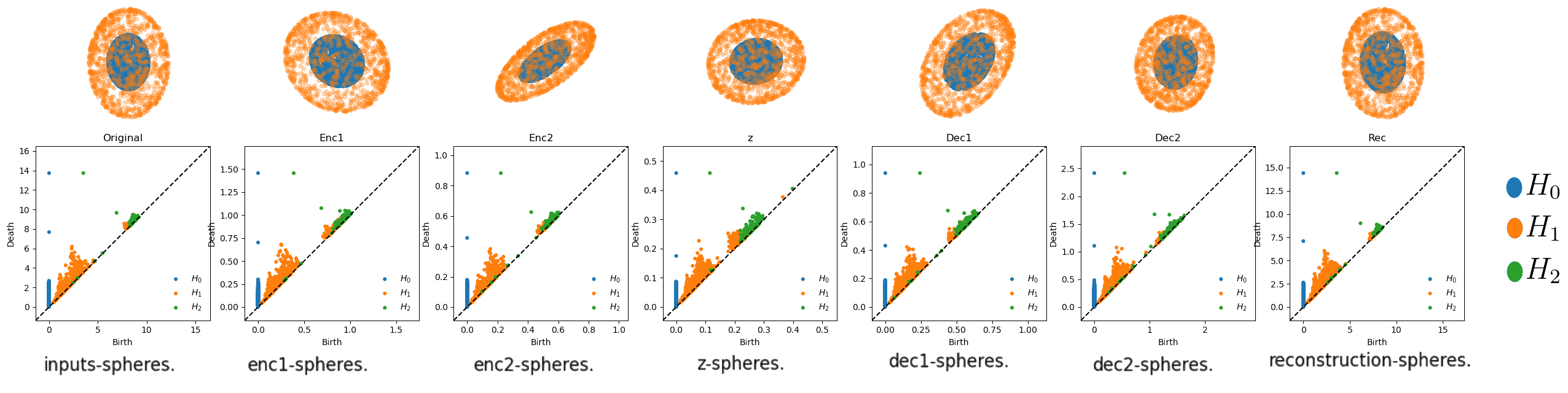}
    \caption{The Tanh latent space information and persistence diagrams up to $H_2$.}
    \label{fig:tanh}
\end{figure}


\begin{table}[H]
\centering
\begin{tabular}{|c|cc|l|}
\hline
\multirow{3}{*}{\textbf{Distance}} & \multicolumn{3}{c|}{\textbf{Time in milliseconds}}        \\ \cline{2-4} 
                                   & \multicolumn{3}{c|}{\textbf{Autoencoder Weight Topology}} \\ \cline{2-4} 
               & \multicolumn{1}{c|}{\textbf{RELU}} & \multicolumn{1}{c|}{\textbf{LRELU}} & \textbf{Tanh} \\ \hline
$WD$           & \multicolumn{1}{c|}{20343.06}      & 25488.97       &16354.78\\ \hline
$Hera\_WD$     & \multicolumn{1}{c|}{10484.34}      & 13395.04       &8896.01 \\ \hline
$SWD_{a=1}$    & \multicolumn{1}{c|}{70.33}         & 97.41          &67.34\\ \hline
$SWD_{a=2}$    & \multicolumn{1}{c|}{119.31}        & 160.05         &104.52\\ \hline
$SWD_{a=4}$    & \multicolumn{1}{c|}{201.45}        & 257.90         &186.33\\ \hline
$SWD_{a=8}$    & \multicolumn{1}{c|}{388.60}        & 519.23         &347.73\\ \hline
$SWD_{a=16}$   & \multicolumn{1}{c|}{706.65}        & 956.24         &655.33\\ \hline
$ETD_{A_1}$    & \multicolumn{1}{c|}{6.00}          & 4.89           &3.85\\ \hline
$ETD_{A_2}$    & \multicolumn{1}{c|}{78.75}         & 87.12          &125.15\\ \hline
$ETD_{A_4}$    & \multicolumn{1}{c|}{142.59.45}     & 165.19         &136.57\\ \hline
$ETD_{A_8}$    & \multicolumn{1}{c|}{306.30}        & 344.08         &272.24\\ \hline
$ETD_{A_{16}}$ & \multicolumn{1}{c|}{630.02}        & 732.73         &570.31\\ \hline
$npETD_{A_1}$    & \multicolumn{1}{c|}{3.46}        & 4.77           &4.70\\ \hline
$npETD_{A_2}$    & \multicolumn{1}{c|}{2.00}        & 2.78           &1.75\\ \hline
$npETD_{A_4}$    & \multicolumn{1}{c|}{3.11}        & 4.27           &2.71\\ \hline
$npETD_{A_8}$    & \multicolumn{1}{c|}{6.19}        & 8.44           &5.33 \\ \hline
$npETD_{A_{16}}$ & \multicolumn{1}{c|}{12.43}       & 16.91          &9.71\\ \hline
$PS$           & \multicolumn{1}{c|}{28.58}         & 40.84          &27.94\\ \hline
Fisher Kernel       & \multicolumn{1}{c|}{3450.52}  & 3526.64        &3056.73\\ \hline
\end{tabular}
\caption{Average time in milliseconds of computing a single distance in $H_0, H_1, H_2$, spanned by activation function and by datasets on the autoencoder experiment.}
\label{tab:autoencoder}
\end{table}
\end{document}